\documentclass[10pt,a4paper,onecolumn,notitlepage]{article}
\usepackage{cmap} 
\usepackage{refcount}
\usepackage[pdfpagemode=None,
pagebackref=true,
bookmarksopen=false,
hyperindex=true,
colorlinks=true,
linkcolor=blue,
citecolor=blue,
pdftitle={Equivalence of tensor products over a category of W*-algebras},
pdfauthor={Ryszard P. Kostecki, Tomasz I. Tylec}]{hyperref}
\usepackage[utf8x]{inputenc}
\usepackage[T1]{fontenc}
\usepackage[polutonikogreek,english]{babel}
\usepackage[all]{xy}
\usepackage{amsmath}
\usepackage{amssymb}
\usepackage{amsthm}
\usepackage{mathrsfs} 
\usepackage{upgreek} 
\usepackage{rpk}
\usepackage{geometry}
\renewcommand*{\backref}[1]{}
\renewcommand*{\backrefalt}[4]{%
  \ifcase #1 %
    \relax
  \or
    $\uparrow$~#2.
  \else
    $\uparrow$~#2.
  \fi%
}

\begin{document}
\thispagestyle{empty}
\selectlanguage{english}
\setcounter{tocdepth}{4}
\setcounter{secnumdepth}{4}

\thispagestyle{empty}
\begin{center}
{\Large\textbf{Equivalence of tensor products over a category of W$^*$-algebras}}{\vspace{0.5cm}}
{\large Ryszard Paweł Kostecki$^{1,2}$, Tomasz Ignacy Tylec$^3$}\\{\vspace{0.4cm}}{\small $^1$ Perimeter Institute for Theoretical Physics, 31 Caroline St North, Waterloo, Ontario N2L 2Y5, Canada}\\{\small $^2$ Institute for Theoretical Physics, University of Warsaw, ul. Pasteura 5, 02-093 Warszawa, Poland}\\{\small $^3$ Institute of Theoretical Physics and Astrophysics, Faculty of Mathematics, Physics and Informatics, University of Gda\'{n}sk, ul. Wita Stwosza 57, 80-308 Gda\'{n}sk, Poland}\\{\vspace{0.25cm}}{\small\texttt{ryszard.kostecki@fuw.edu.pl}, \texttt{ttylec@gmail.com}}\\{\vspace{0.4cm}}{December 20, 2017}{\small\\\ \\} 
\end{center}

\begin{abstract}
\noindent We prove the equivalence of two tensor products  over a category of W$^*$-algebras with normal (not necessarily unital) $^*$-homomorphisms, defined by Guichardet and Dauns, respectively. This structure differs from the standard tensor product construction by Misonou--Takeda--Turumaru, which is based on weak topological completion, and does not have a categorical universality property.
\end{abstract}

\section{Introduction}

The finite dimensional sector of von Neumann's Hilbert space based framework for quantum mechanics \cite{vonNeumann:1932:grundlagen} was reformulated in \cite{Abramsky:Coecke:2004,Selinger:2004,Selinger:2007,Abramsky:Coecke:2009} in terms of symmetric monoidal categories equipped with further structural properties. However, the extension of categorical foundations for quantum mechanics to the infinite dimensional regime (thus, category-theoretisation of the original object of concern of von Neumann) remains an open problem. In this paper we prove that two alternative constructions of a tensor product over a category $\Wsn $ of W$^*$-algebras with normal $^*$-homomorphisms are equivalent. One of them (denoted here by $\otimes_G$) was introduced by Guichardet \cite{Guichardet:1966}, another one (denoted here by $\underline{\otimes}$) was introduced by Dauns \cite{Dauns:1972}. On the other hand, the most popular tensor product structure over W$^*$-algebras is the one defined by Misonou, Takeda, and Turumaru \cite{Misonou:1954,Takeda:1954,Turumaru:1954} as the weak closure of the algebraic tensor product over a Hilbert space defined by tensor product of faithful normal representations of composite W$^*$-algebras. However, this tensor product structure (denoted here as $\overline{\otimes}$) lacks categorical universality property and, furthermore, it is not equivalent with $\underline{\otimes}$ if the composite W$^*$-algebras are not nuclear. This leads us to suggest the symmetric monoidal category $(\Wsn,\underline{\otimes},\CC)$ as a point of departure for further category theoretic axiomatisation of infinite-dimensional quantum mechanics. 

In Section \ref{tensor.products.Wstar.section} we recall the basic facts and definitions of the tensor products over W$^*$-algebras. In Section \ref{cat.tensor.Wstar.section} we present Guichardet's construction, and prove that it is equivalent with Dauns'.
\section{Analytic tensor products of W$^*$-algebras\label{tensor.products.Wstar.section}}
For any two infinite dimensional W$^*$-algebras $\N_1$ and $\N_2$ there exist different inequivalent tensor product structures $\otimes$, allowing to form a compound W$^*$-algebra $\N_1\otimes\N_2$. The variety of these structures arises from different possible ways of introducing a topology on the algebraic tensor product of $\N_1$ and $\N_2$ which makes it into a W$^*$-algebra. If $\N_1$ or $\N_2$ is finite dimensional, then all those tensor product structures coincide. In this section we will review the results of the general theory that allows to deal with the generic infinite dimensional case.

Let $X$ be a Banach space, $M$ a closed subspace
of $X$, and $Y$ a subset of $X$. A Banach dual space of $X$ will be denoted $X^\banach$. Then $M$ and $X/M$ are also Banach spaces. An \df{annihilator} of $Y$ in $X^\banach$ is defined as 
\begin{equation}
	Y^\bot:=\{z\in X^\banach\mid z(x)=0\;\forall x\in X\}
\end{equation}
These objects satisfy \cite{Dieudonne:1942}
\begin{equation}
	M^\banach\iso X^\banach/M^\bot,\;\;\;(X/M)^\banach\iso M^\bot.
\label{annihilator.banach.equation}
\end{equation}
If $X$ is a Banach dual space of some Banach space $Z$, then $Z$ is called a \df{predual} of $X$, and is denoted by $X_\star$.


All $C^*$-algebras in this text are assumed to contain a unit $\II$. The weakly-$\star$ continuous linear maps between W$^*$-algebras will be called \df{normal}.  In particular, as follows from \cite{Bratteli:Robinson:1979}, Prop. 2.4.2 and 2.4.3, the left and right multiplication maps $a\mapsto ab, a\mapsto ba$ are weak-$\star$ continuous.

For any Banach space $X$, there is a canonical isometric embedding map $j_X:X\ra X^\banach{}^\banach$, defined by \cite{Hahn:1927}
\begin{equation}
	(j_X(x))(\phi):=\phi(x)\;\;\forall\phi\in X^\banach\;\forall x\in X.
\end{equation}
If $\C$ is a $C^*$-algebra, then $\C^\banach{}^\banach$ is a W$^*$-algebra, called a \df{universal enveloping W$^*$-algebra}, while $j_\C$ is a $^*$-isomorphism onto a weak-$\star$ dense $^*$-subalgebra of $\C^\banach{}^\banach$ \cite{Sherman:1952,Takeda:1954:conjugate}.

Given two vector spaces $X$ and $Y$ over $\KK\in\{\RR,\CC\}$, $X\boxtimes Y$ will denote the algebraic tensor product of $X$ and $Y$, which is again a vector space over $\KK$. For any vector space $X$ over $\KK$, $\X\boxtimes\KK\iso X$. Given Banach spaces $X$ and $Y$, a norm $\n{\cdot}$ on $X\boxtimes Y$ is called a \df{cross norm} if{}f \cite{Schatten:1943}
\begin{equation}
	\n{x\otimes y}=\n{x}_X\n{y}_Y\;\;\;\forall(x,y)\in X\times Y.
\end{equation}
The completion of $X\boxtimes Y$ in the topology of $\n{\cdot}$ is denoted $X\otimes_{\n{\cdot}}Y$.

For any $C^*$-algebras $\C_1$ and $\C_2$, $\C_1\boxtimes\C_2$ is a $^*$-algebra \cite{Turumaru:1952}. A seminorm $p$ on $\C_1\boxtimes\C_2$ that satisfies $p(x^*x)=p(x)^2$ $\forall x\in\C_1\boxtimes\C_2$ is called a \df{$C^*$-seminorm}. A norm $\n{\cdot}$ on $\C_1\boxtimes\C_2$ that satisfies $\n{x^*x}=\n{x}^2$ $\forall x\in\C_1\boxtimes\C_2$ is called a \df{$C^*$-norm}. Each $C^*$-norm is a cross norm and satisfies  $\n{xy}\leq\n{x}\n{y}$ $\forall x,y\in\C_1\boxtimes\C_2$. A completion of $\C_1\boxtimes\C_2$ in the topology of a $C^*$-norm $\n{\cdot}$ is a $C^*$-algebra, denoted $\C_1\otimes_{\n{\cdot}}\C_2$. The definition of $C^*$-norm $\n{\cdot}$ does not imply the isometric isomorphism $\C_1\otimes_{\n{\cdot}}\C_2\iso\C_2\otimes_{\n{\cdot}}\C_1$ (see \cite{Effros:Lance:1977} for an example). For any $C^*$-algebras $\C_1$ and $\C_2$, if $\phi_1\in\C_1^{\banach+}$ and $\phi_2\in\C_2^{\banach+}$, then $\phi_1\boxtimes\phi_2$ is continuous with respect to any $C^*$-norm on $\C_1\boxtimes\C_2$.

A $C^*$-norm \cite{Guichardet:1965}
\begin{align}
	\n{x}^{C^*}_{\mathrm{max}}:&=\sup\{p(x)\mid p\;\mbox{is a }C^*\mbox{-norm on }\C_1\boxtimes\C_2\}\\
	&=\sup\{p(x)\mid p\;\mbox{is a }C^*\mbox{-seminorm on }\C_1\boxtimes\C_2\},
\end{align}
is `projective' in the following sense: for any $C^*$-algebras $\C_1$ and $\C_2$, and any closed two sided ideal $\mathcal{I}_1\subseteq\C_1^\sa$,
\begin{equation}
	(\C_1/\mathcal{I}_1)\omax\C_2\iso(\C_1\omax\C_2)/(\mathcal{I}_1\omax\C_2).
\end{equation}
If $I$ is any closed two sided ideal in $\C_1\omax\C_2$ such that $(\C_1\boxtimes\C_2)\cap I=\{0\}$, then the quotient norm on $(\C_1\omax\C_2)/I$ is a $C^*$-norm on $\C_1\boxtimes\C_2$. It satisfies the following universal property: let $\C_1$, $\C_2$, $\C$ be $C^*$-algebras, if $\varsigma_i:\C_i\ra \C$, $i\in\{1,2\}$, are $^*$-homomorphisms with pointwise commuting ranges (i.e., for $x\in\varsigma_1(\C_1)$ and $y\in\varsigma(\C_2)$ one has $xy=yx$), then there exists a unique $^*$-homomorphism $\varsigma:\C_1\omax\C_2\ra\C$ such that $\varsigma(x_1\otimes x_2)=\varsigma_1(x_1)\varsigma_2(x_2)$ and $\varsigma(\C_1\omax\C_2)$ is equal to the $C^*$-subalgebra of $\C$ generated by $\varsigma_1(\C_1)$ and $\varsigma_2(\C_2)$. An alternative characterisation of $\n{\cdot}^{C^*}_{\mathrm{max}}$ was given in \cite{Guichardet:1969}. For $\C_1,\C_2,\C$ and $\varsigma_i$ as above, and for $m:\C\otimes\C\ni x\otimes y\mapsto xy\in\C$,
\begin{equation}
	\n{x}^{C^*}_{\mathrm{max}}:=\sup\left\{\n{m\circ(\varsigma_1\boxtimes\varsigma_2)(x)}\mid\C,\;\varsigma_1,\varsigma_2\right\}.
\end{equation}
The universal property of $\n{\cdot}^{C^*}_{\mathrm{max}}$ can be restated as a commutative diagram
\begin{equation}
\xymatrix{%
&\C_1\omax\C_2\ar@{..>}[dd]&\\
\C_1\ar[ur]^{w_1}\ar[dr]_{\varsigma_1}&&\C_2\ar[ul]_{w_2}\ar[dl]^{\varsigma_2}\\
&\C&
}
\end{equation}
where $w_1$ and $w_2$ are $^*$-homomorphisms that are required to satisfy $[w_1(x_1),w_2(x_2)]=0$.


Consider a $C^*$-norm defined by \cite{Turumaru:1952,Turumaru:1953}
\begin{equation}
	\n{x}^{C^*}_{\mathrm{min}}:=\n{(\pi_1\boxtimes\pi_2)(x)}_{\BBB(\H_1\otimes\H_2)},
\end{equation}
where $(\H_1,\pi_1)$ and $(\H_2,\pi_2)$ are faithful representations of $\C_1$ and $\C_2$, respectively. This definition is independent of the choice of particular representations \cite{Misonou:1954,Takeda:1954}. It is `injective' in the following sense: if $\C_3$ and $\C_4$ are $C^*$-subalgebras of $\C_1$ and $\C_2$, respectively, then the embedding $\C_3\boxtimes\C_4\subseteq\C_1\boxtimes\C_2$ extends to an isometric embedding $\C_3\omin\C_4\subseteq\C_1\omin\C_2$. (Because the notions of `projective' and `injective' tensor products for Banach spaces do not coincide with those for $C^*$-algebras, we will avoid using these adjectives.) Every $C^*$-norm $\n{\cdot}$ satisfies \cite{Takesaki:1964}
\begin{equation}
	\n{x}^{C^*}_{\mathrm{min}}\leq
	\n{x}\leq
	\n{x}^{C^*}_{\mathrm{max}}
	\;\;\;\forall x\in\C_1\boxtimes\C_2,
\end{equation}
with lower bound attained if{}f $\C_1$ or $\C_2$ is commutative \cite{Takesaki:1958} (this implies $\mathrm{C}(X)\omin\mathrm{C}(Y)\iso\mathrm{C}(X\times Y)$ for compact Hausdorff spaces $X$ and $Y$ \cite{Turumaru:1952}). Thus, the set of all $C^*$-norms on $\C_1\boxtimes\C_2$ is a complete lattice. A $C^*$-algebra $\C_1$ is called \df{nuclear} if{}f 
\begin{equation}
	\n{x}^{C^*}_{\mathrm{min}}=\n{x}^{C^*}_{\mathrm{max}}\;\;\;\forall x\in\C_1\boxtimes\C_2
\end{equation}
holds for any $C^*$-algebra $\C_2$ \cite{Takesaki:1964,Lance:1973}. 
All finite dimensional and all commutative $C^*$-algebras are nuclear. If $\H$ is an infinite dimensional separable Hilbert space, then $\BH$ is not nuclear \cite{Wassermann:1976}. 

Given W$^*$-algebras $\N_1$ and $\N_2$, and a $C^*$-norm $\n{\cdot}$, the $C^*$-algebra $\N_1\otimes_{\n{\cdot}}\N_2$ is not necessary a W$^*$-algebra. However, one can prove the following lemma.

\begin{lemma}\
Given W$^*$-algebras $\N_1$ and $\N_2$, a $C^*$-norm $\n{\cdot}$ on $\N_1\boxtimes\N_2$, let $Y$ be a closed subspace of $(\N_1\otimes_{\n{\cdot}}\N_2)^\banach$ that is invariant under left and right multiplication by the elements of $\N_1\otimes_{\n{\cdot}}\N_2$. Then
\begin{equation}
	\N_1\otimes_{\n{\cdot},Y}\N_2:=(\N_1\otimes_{\n{\cdot}}\N_2)^\banach{}^\banach/Y^\bot\iso Y^\banach
\label{Wstar.tensor.construction}
\end{equation}
is a W$^*$-algebra.
\end{lemma}
\begin{proof}
From definition, $(\N_1\otimes_{\n{\cdot}}\N_2)^\banach{}^\banach$ is a W$^*$-algebra and $Y^\bot$ is a two sided ideal in it. Last equation follows from the general Banach space property \eqref{annihilator.banach.equation}.
\end{proof}

The special case of the above construction has been used in \cite{Sakai:1971} for $\n{\cdot}=\n{\cdot}^{C^*}_{\mathrm{min}}$ and $Y=(\N_1)_\star\otimes_{(\n{\cdot}^{C^*}_{\mathrm{min}})^\banach}(\N_2)_\star=:(\N_1)_\star\overline{\otimes}_\star(\N_2)_\star$. The resulting tensor product W$^*$-algebra, $\N_1\overline{\otimes}\N_2$, is equivalent with the tensor product of $\N_1$ and $\N_2$ defined in \cite{Misonou:1954} as a von Neumann subalgebra of $\BBB(\H\otimes\K)$ that is a weak closure of $\pi_1(\N_1)\boxtimes\pi_2(\N_2)$ on $\H_1\otimes\H_2$, where $(\H_1,\pi_1)$ and $(\H_2,\pi_2)$ are faithful normal representations of $\N_1$ and $\N_2$, respectively. $\N_1\overline{\otimes}\N_2$ is a weakly-$\star$ dense subspace of $((\N_1)_\star\overline{\otimes}_\star(\N_2)_\star)^\star$, and canonical embedding of the former into the latter is a $^*$-isomorphism. 

Another special case of the construction \eqref{Wstar.tensor.construction} was proposed in \cite{Dauns:1972,Dauns:1978} for $\n{\cdot}=\n{\cdot}^{C^*}_{\mathrm{max}}$ and $Y=:(\N_1)_\star\underline{\otimes}_\star(\N_2)_\star$ defined as a set of all $\phi\in(\N_1\omax\N_2)^\banach$ satisfying $\phi(x\otimes\cdot)\in(\N_2)_\star$ and $\phi(\cdot\otimes y)\in(\N_1)_\star$ $\forall(x,y)\in\N_1\times\N_2$. 

The tensor product W$^*$-algebra, $\N_1\underline{\otimes}\N_2$, satisfies the following property: if $\N_i$, $i\in\{1,\ldots,4\}$, are W$^*$-algebras, $\alpha_1:\N_1\ra\N_3$, $\alpha_2:\N_2\ra\N_4$ are weak-$\star$ continuous $^*$-homomorphisms, then there exists a unique weak-$\star$ continuous $^*$-homomorphism $\alpha:\N_1\underline{\otimes}\N_2\ra\N_3\underline{\otimes}\N_4$ such that $\alpha(x\otimes y)=\alpha_1(x)\otimes\alpha_2(y)$ \cite{Dauns:1972}. The analogous result holds for $\overline{\otimes}$ and weak-$\star$ continuous $^*$-homomorphisms of W$^*$-algebras \cite{Misonou:1954,Takeda:1954,Turumaru:1954}.

The tensor product $(\N_1)_\star\overline{\otimes}_\star(\N_2)_\star$ can be constructed as a projective tensor product of operator spaces \cite{Effros:Ruan:1991,Blecher:Paulsen:1991}, and it satisfies $(\N_1\overline{\otimes}\N_2)_\star\iso(\N_1)_\star\overline{\otimes}_\star(\N_2)_\star$ \cite{Effros:Ruan:1990}. On the other hand, the tensor product $\underline{\otimes}_\star$ satisfies $(\N_1\underline{\otimes}\N_2)_\star\iso(\N_1)_\star\underline{\otimes}_\star(\N_2)_\star$ \cite{Dauns:1972}. 

\section{Categorical tensor products of W$^*$-algebras\label{cat.tensor.Wstar.section}}
Guichardet \cite{Guichardet:1966} introduced a category $\Wsn$ of W$^*$-algebras 
and normal (not necessarily unital) $^*$-ho\-mo\-morph\-isms between them.
A degenerate algebra $\O = \{0\}$,
consisting of only one element, is considered as an object of $\Wsn$.
Clearly, this is a terminal object of $\Wsn$.
It is also an initial object of $\Wsn$, 
since the only linear map $\O \to \A$ is $0 \mapsto 0$.
Consequently, $\O$ is a zero object
and $\Wsn$ has zero morphisms, i.e.\ for any $\A, \B$
there is a unique $0_{\A, \B}\in\Hom(\A, \B)$
defined by $\A \overset!\ra \O \overset!\ra \B$. (In case of category with unital morphisms, 
we do not have the zero object since 
since there is no unital morphism from the $\O$ algebra
to any other algebra.)

For a countable family $\{\A_i\}_{i\in I}$ of W$^*$-algebras 
acting on Hilbert spaces $\H_i$
we define their \emph{product} (cf.~\cite{Dixmier:1957}) $\A = \prod_{i\in I} \A_i$
as a von Neumann algebra which elements are sequences $(a_i)_{i\in I}, a_i\in \A_i$,
such that $\sup_{i\in I}\left\{\n{a_i}\right\} < \infty$,
acting on a direct sum $\H = \bigoplus_{i\in I} \H_i$
in the following way: 
\begin{equation*}
    x = (x_i)_{i\in I} \mapsto a x = (a_ix_i)_{i\in I}.
\end{equation*}
Clearly it is a product in the $\Wsn$ category \cite{Guichardet:1966}:
for any family $u_i\colon \B \to \A_i$ we define $u\colon \B \to \prod_i\A_i$
by $b\mapsto u(b) = (u_i(b))$; 
then $u_i = p_i\circ u$, where $p_i\colon\prod_i\A_i\to\A_i$
are canonical projections.
Moreover, it satisfies following universality property:

\begin{proposition}[\cite{Guichardet:1966}, remark 3.2]
    Let $u_i\colon \A_i\to \B$ be a family of morphisms in $\Wsn$
    such that $u_i(x_i)u_j(x_j) = 0$ for $i\neq j$.
    Then, there exists a unique morphism $u\colon\prod_i \A_i \to \B$
    such that $u_i = u \circ s_i$,
    where $s_i\colon\A_i \to \prod_i \A_i$ are canonical injections.
\end{proposition}

Guichardet defined the tensor product in this category by means of the following universal property.

\begin{definition}\label{guichardet.tensor.def}
Let $\N_1,\N_2,\N$ be W$^*$-algebras, let $w_1:\N_1\ra\N$ and $w_2:\N_2\ra\N$ be normal $^*$-homomorphisms such that
\begin{enumerate}
\item[1)] $[w_1(x_1),w_2(x_2)]=0$ $\forall x_1\in\N_1$ $\forall x_2\in\N_2$,
\item[2)] for any W$^*$-algebra $\M$ and any normal $^*$-homomorphisms $t_1:\N_1\ra\M$ and $t_2:\N_2\ra\M$ such that $[t_1(x_1),t_2(x_2)]=0$ $\forall x_1\in\N_1$ $\forall x_2\in\N_2$ there exists a unique normal $^*$-homomorphism $t:\N\ra\M$ such that the following diagram commutes
\begin{equation}
\xymatrix{%
&\N\ar@{..>}[dd]^{t}&\\
\N_1\ar[ur]^{w_1}\ar[dr]_{t_1}&&\N_2\ar[ul]_{w_2}\ar[dl]^{t_2}\\
&\M&
}
\end{equation}
\end{enumerate}
Then $\N$ is denoted as $\N_1\otimes_G\N_2$.
\end{definition}
\begin{proof}
    We have to show that for any pair $\N_1, \N_2$ of W$^*$-algebras
    there exist a $\N_1\otimes_G\N_2$. 
    Let us denote by:
    \begin{align*}
        r_1(a) &= a \otimes \II, \qquad a\in \N_1,\\
        r_2(b) &= \II \otimes b, \qquad b\in \N_2.
    \end{align*}
    maps $r_i\colon\N_i\to \N_1\boxtimes\N_2$.
    We say that $^*$-homomorphism $u\colon \N_1\boxtimes\N_2\to\M$,
    where $\M$ is some W$^*$-algebra
    and $u(\N_1\boxtimes\N_2)$ is weakly-$\star$ dense in $\M$,
    is \emph{normal} whenever both $u\circ r_1, u\circ r_2$
    are normal as a maps $\N_i\to\M$.
    Further, we say that two such normal maps:
    $u\colon\N_1\boxtimes\N_2\to \A$,
    $v\colon\N_1\boxtimes\N_2\to \B$,
    where $\A, \B$ are two arbitrary W$^*$-algebras,
    are equivalent, whenever there exist
    a normal isomorphism $f\colon \A\to\B$
    such that $v = f\circ u$.
    It can be shown (\cite{Guichardet:1966}, Lemma 4.2)
    that equivalence classes of such maps form a set.
    Observe also that there are always at least two such classes,
    represented by maps:
    \begin{equation*}
        p_1(a\otimes b) = a,\qquad p_2(a\otimes b) = b,\qquad 
        \text{and extended by linearity;}
    \end{equation*}
    (clearly $p_i\circ r_j$ are normal).
    Now let us a take one representant $u_j\colon\N_1\boxtimes\N_2\to \M_i$
    out of each of above equivalence classes.
    Denote by $g\colon\N_1\boxtimes\N_2\to\N\subset\prod_i\M_i$
    the $^*$-homomorphism made out of $(u_j)$ 
    and the weak-$\star$ closure of $(\prod_j u_j)(\N_1\boxtimes\N_2)$ in $\prod_i\M_i$.
    By definition it is a W$^*$-algebra.
    Denote by $w_i = g\circ r_i$, for $i=1,2$.

    Let us define a map $u(a\otimes b) = t_1(a)t_2(b)$ and extend by linearity
    to $\N_1\boxtimes\N_2\to\M$. Observe that $u\circ r_1 = t_1(a)t_2(\II)$
    is a composition of two weakly-$\star$ continuous maps ($t_1$ and right
    multiplication by $t_2(\II)$), thus $u\circ r_1$ is also weak-$\star$
    continuous. Analogously $u\circ r_2$ is weak-$\star$ continuous.
    Consequently, $u$ is a normal map. As such, we know that there exists $j$
    such that $u_j$ is equivalent to $u$, i.e.~there exists $f\colon \M_i\to\M$
    such that $u = f\circ u_j$. As a result $t=f\circ p_j$:

    \begin{equation}
    \xymatrix{%
    &\N\ar@{..>}[d]^{p_j}&\\
    \N_1\ar[ur]^{w_1}\ar[dr]_{t_1}&\M_j\ar@{..>}[d]^{f}&\N_2\ar[ul]_{w_2}\ar[dl]^{t_2}\\
    &\M&
    }
    \end{equation}
    Not that although isomorphism $f$ does not have to be unique, the whole
    construction of tensor product is up to isomorphism (we choose $u_j$ from
    equivalence classes).
    This completes the proof of universality.
\end{proof}

Dauns \cite{Dauns:1972} introduced another tensor product in $\Wsn$, denoted in Section \ref{tensor.products.Wstar.section} as $\N_1\underline{\otimes}\N_2$. He showed that $\N_1\underline{\otimes}\N_2$ is characterised by the universal property analogous to one given in Definition \ref{guichardet.tensor.def}, but specified in the category $\Wsun$ of W$^*$-algebras and normal unital $^*$-homomorphisms. Dauns showed also that $(\Wsn,\underline{\otimes},\CC)$ and $(\Wsun,\underline{\otimes},\CC)$ are symmetric monoidal categories. However, the relationship between the tensor products $\otimes_G$ and $\underline{\otimes}$ was left unspecified, so let us fill this gap.

\begin{proposition}
For any $\N_1,\N_2\in\Ob(\Wsn)$ there is a normal unital $^*$-isomorphism $\N_1\otimes_G\N_2\iso\N_1\underline{\otimes}\N_2$.
\end{proposition}
\begin{proof}
    Universality of $\N_1\underline{\otimes}\N_2$ (\cite{Dauns:1972}, 4.8) means that
    for any unital $^*$-homomorphisms $\alpha\colon\N_1\to\M$ and $\beta\colon\N_2\to\M$,
    such that $[\alpha(\N_1),\beta(\N_2)] = 0$
    there exists a unique unital $^*$-homomorphism such that
    the following diagram commutes:
    \begin{equation*}
        \xymatrix{
            & \M &\\
            \N_1 \ar[r]^{v_1}\ar[ur]^{\alpha}
            & \N_1\underline{\otimes}\N_2\ar@{..>}[u]^{f}
            & \ar[l]_{v_2}\ar[ul]_{\beta} \N_2\\
        }
    \end{equation*}
    where $v_1, v_2$ are natural inclusions of $\N_1,\N_2$ into $\N_1\underline{\otimes}\N_2$. 

    Let $g, (u_j), w_i, r_i$ be defined as in the proof of Def.~\ref{guichardet.tensor.def}.
    From the $^*$-homomorphism property we have that for
    \begin{equation*}
      u_j(\II)u_j(a) = u_j(a) = u_j(a)u_j(\II)\qquad\forall a\in\N_1\boxtimes\N_2,
    \end{equation*}
    since the image of $u_j$ is weakly-$\star$ dense in $\M_j$
    and since $u_j(\II)u_j(a)$ is weakly-$\star$ continuous
    (composition of weakly-$\star$ continuous $u_j$ and left multiplication),
    we can extend this equality by continuity to $\M_j$.
    Consequently, $u_j(\II) = \II$ and thus
    $g(\II) = \II$ and $w_i = g\circ r_i$ are unital. 
    Then the diagram
    \begin{equation*}
        \xymatrix{
            & \N_1\underline{\otimes}\N_2 &\\
            \N_1 \ar[ur]^{v_1}\ar[r]^{w_1}\ar[dr]_{v_1} 
            & \N_1\otimes_G\N_2\ar@{..>}[u]^{f}
            & \ar[ul]_{v_2}\ar[l]_{w_2}\ar[dl]^{v_2}\N_2\\
            & \N_1\underline{\otimes}\N_2\ar@{..>}[u]^{h}
        }
    \end{equation*}
    commutes,
    where existence of unique $h$ follows from universality of $\underline{\otimes}$
    and existence of unique $f$ follows from universality of $\otimes_G$.
    From universality of $\underline{\otimes}$ the whole diagram yields that
    $f\circ h = \id_{\N_1\underline{\otimes}\N_2}$.
    The other way follows analogously.
\end{proof}

\noindent\textbf{Acknowledgments.} RPK would like to thank L.Hardy, R.Kunjwal, and J.Lewandowski for hosting him as a scientific visitor. This research was supported in part by Perimeter Institute for Theoretical Physics. Research at Perimeter Institute is supported by the Government of Canada through Industry Canada and by the Province of Ontario through the Ministry of Research and Innovation.

{\small
\section*{References}
\begingroup
\raggedright
\renewcommand\refname{

\endgroup        
}
\end{document}